\documentclass[journal,12pt,draftcls,onecolumn]{IEEEtran}

\usepackage{graphicx}      

\usepackage{amsmath} 
\usepackage{amssymb}  
\usepackage{amsthm}

\usepackage{algorithm}
\usepackage{algpseudocode}

\usepackage{cite}        

\newtheorem{definition}{Definition}

\newtheorem{remark}{Remark}
\newtheorem{example}{Example}
\newtheorem{thm}{Theorem}
\newtheorem{prop}{Proposition}

\newcommand{\abs}[1]{\left\vert#1\right\vert}
\newcommand{\norm}[1]{\left\Vert#1\right\Vert}
\newcommand{\set}[1]{\left\{#1\right\}}
\newcommand{\Real}{\mathbb R}
\newcommand{\Z}{\mathbb Z}
\newcommand{\U}{\mathcal U}
\renewcommand{\u}{\mathbf u}
\newcommand{\T}{\mathcal T}
\newcommand{\Q}{\mathcal Q}

\begin{document}

\title{\Large\bfseries Computing finite abstractions with robustness margins via local reachable set over-approximation}

\author{Yinan Li, Jun Liu, and Necmiye Ozay%
\thanks{This work is supported, in part, by EU FP7 Grant PCIG13-GA-2013-617377 and NSF grant CNS-1446298. This is an extended version of the paper \cite{li2015computing} to appear in the Proceedings of the 2015 IFAC Conference on Analysis and Design of Hybrid Systems (ADHS), Atlanta, GA, USA, October 14-16, 2015.}
\thanks{Yinan Li and Jun Liu are with the Department of Automatic Control and System Engineering, University of Sheffield, Sheffield, S1 3JD, UK (e-mails: yli107@sheffield.ac.uk; j.liu@sheffield.ac.uk).}
\thanks{Necmiye Ozay is with the Department of Electrical Engineering and Computer Science, University of Michigan, Ann Arbor, MI 48109, USA (e-mail: necmiye@umich.edu).}}

\maketitle

\begin{abstract}                
This paper proposes a method to compute finite abstractions that can be used for synthesizing robust hybrid control strategies for nonlinear systems. Most existing methods for computing finite abstractions utilize some global, analytical function to provide bounds on the reachable sets of nonlinear systems, which can be conservative and lead to spurious transitions in the abstract systems. This problem is even more pronounced in the presence of imperfect measurements and modelling uncertainties, where control synthesis can easily become infeasible due to added spurious transitions. To mitigate this problem, we propose to compute finite abstractions with robustness margins by over-approximating the local reachable sets of nonlinear systems. We do so by linearizing the nonlinear dynamics into linear affine systems and keeping track of the linearization error. It is shown that this approach provides tighter approximations and several numerical examples are used to illustrate of effectiveness of the proposed methods.

\end{abstract}

\begin{IEEEkeywords}
Nonlinear systems, temporal logic, control synthesis, reachable set computation.
\end{IEEEkeywords}


\section{Introduction}

Construction of finite abstractions for nonlinear systems is a critical step when applying abstraction-based approaches to hybrid control synthesis \cite{Alur00}. Such approaches have gained popularity over the past few years for their ability to handle control problems for complex dynamical systems from high-level, rigorous specifications (see, e.g., piecewise affine systems \cite{Kloetzer08,YordanovTCBB12}, polynomial and nonlinear switched systems \cite{OzayLPM13,LiuOTM13}.
The underlying principle of such approaches is to search for a controller in a finite abstraction of the original continuous system, leveraging formal synthesis techniques developed in computer science. As a result, the fidelity of finite abstractions has a significant influence on the result of control synthesis.

Symbolic models that are approximately similar or bisimilar to continuous-time nonlinear systems have been proposed and studied extensively \cite{PolaGT08,GirardPT10,ZamaniPMT12,TabuadaBook09}, which provide concrete means for computing finite approximate models often based on state-space discretization. For example, the symbolic models proposed in \cite{PolaGT08} and \cite{GirardPT10} are based on approximate bisimulation relations, which require incremental input-to-state stability \cite{Angeli02} of the original system. The work by \cite{ZamaniPMT12} later relaxes the stability  requirement and constructs symbolic models that are essentially approximately alternatingly similar to the original system. Such symbolic models are nondeterministic and the computation of transitions relies on a global, analytical function provided by the incremental forward completeness of dynamics \cite{ZamaniPMT12}.

When dynamical systems are affected by imperfections such as measurement errors, delays, and disturbances, synthesis of robust control strategies using abstraction-based approaches becomes important. Motivated by this, the work by \cite{LiuO14} introduces a notion of finite abstractions that are equipped with additional robustness margins to account for imperfections in measurements and/or models. These margins also lead to added nondeterminism in the abstractions.

To increase the fidelity of the nondeterminitic finite abstractions, one needs to reduce the number of spurious transitions in the abstractions. One way to do so is to compute tighter approximations of the local reachable sets for nonlinear systems. While local reachable set computation has been used for nonlinear system analysis and verification (see, e.g., \cite{Althoff10,Althoff14}), we use it here to compute finite abstractions for robust control synthesis. More specifically, we linearize the nonlinear dynamics and keep track of the linearization errors. Robustness margins are incorporated in the set of initial conditions used for computing local reachable sets. This allows us to use margins that are are state-dependent and take into account variations in local dynamics. One major advantage of the proposed approach is that it provides much less conservative abstractions, compared with existing approaches.

\textbf{Notation}: let $\mathbb{Z}$ be the set of integers and $\mathbb{N}$ be the set of all nonnegative integers; $\mathbb{R}$ represents the set of all real numbers; $\mathbb{R}_{\geq 0}$ and $\mathbb{R}_{>0}$ are the sets of all nonnegative and all positive real numbers, respectively; $\mathbb{R}^n$ denotes the $n$-dimensional Euclidean space; $\Z^n$ denotes the $n$-dimensional integer lattice (the set of vectors in $\mathbb{R}^n$ whose components are all integers); given a vector $x=(x_1,\cdots,x_n)$ in $\Real^n$, let $\abs{x}=(\abs{x_1},\cdots,\abs{x_n})$, i.e., the vector obtained by taking entrywise absolute value of $x$; given two vectors $x=(x_1,\cdots,x_n)$ and $y=(y_1,\cdots,y_n)$, $x\le y$ means $x_i\le y_i$ for all $i\in \set{1,\cdots,n}$ ($x<y$, $x>y$, and $x\ge y$ are similarly defined) and $x\circ y$ indicates the entrywise product, i.e., $x\circ y:=(x_1y_1,\cdots,x_ny_n)$; a vector $x\in\Real^n$ is said to be positive if $x>0\in\Real^n$ and nonnegative if $x\ge 0\in\Real^n$; let $\mathbb{R}^n_{>0}$ and $\mathbb{R}^n_{\ge 0}$ denote the set of positive and nonnegative vectors in $\Real^n$; given vectors $\delta\in \mathbb{R}^n_{\ge 0}$ and $x\in\Real^n$, define $B_{\delta}(x):=\set{x'\in\Real^n:\,\abs{x'-x}\le\delta}$, a hyper-rectangular box centred at $x$; $\mathcal{B}_\delta(0)$ is written as $\mathcal{B}_\delta$ for short; given $\eta \in \mathbb{R}^n_{\ge 0}$, define $[\mathbb{R}^n]_{\eta}:=\set{\eta\circ k\in\Real^n:\,k\in\mathbb{Z}^n}$ to be a hyper-rectangular grid with granularity parameter $\eta$; given a set $S \subseteq \mathbb{R}^n$ and a vector $\eta \in \mathbb{R}^n_{\ge 0}$, define $[S]_{\eta}:= S\cap [\mathbb{R}^n]_{\eta}$ to be the set of all grid points in $S$; given two sets $X\subseteq \Real^n$ and $Y\subseteq \Real^n$, $X \oplus Y$ denotes their Minkowski addition defined as $X \oplus Y:= \{ x+y |\; x \in X,\, y \in Y\}$; given a function $f$, dom$(f)$ denotes its domain.

\section{Problem formulation}\label{sec:prob}

\subsection{Continuous-time control system}

We consider a continuous-time control system described by a tuple $\T:=(X,X_0,U,f,\Pi,L)$, whose execution is governed by the ordinary differential equation with inputs
\begin{equation}
\dot{x}(t)=f(x(t),\,u(t)),\label{eq:1}
\end{equation}
where $t \in \mathbb{R}_{\geq 0}$, $x(t) \in X \subseteq \mathbb{R}^n$ is the system state, $x(0)\in X_0\subseteq \mathbb{R}^n$ is the initial state, and $u(t) \in U \subseteq \mathbb{R}^m$ is the control input. A measurable locally essentially bounded function defined on $[0,\tau]$ taking values in $U$ is called a \emph{control signal} of duration $\tau$. Let $\mathcal{U}$ be the set of all control signals with arbitrary but finite duration. The vector field $f:\mathbb{R}^n \times \mathbb{R}^m \to \mathbb{R}^n$ is a continuous function that fulfills the basic conditions  (see, e.g., \cite{KhalilBook02}) for existence and uniqueness of solutions:  given $x_0 \in X$, $T \in \mathbb{R}_{\geq 0}$, and a control signal $\mathbf{u}$ of duration $T$, there exists a unique solution, denoted by $\xi(t,\,x_0,\,\mathbf{u})$, that satisfies (\ref{eq:1}) for $t\in[0,T]$ and the initial condition $x(0)=x_0$. The labeling function $L:X \to 2^\Pi$ is function that maps a state of $\T$ to a set of propositions in $\Pi$ that hold true at this state.

\subsection{LTL control synthesis problem}

The desired system behaviors for $\T$ are specified using linear temporal logic (LTL). LTL is able to express a combination of safety, reachability, invariance properties. It is built upon the set of atomic propositions $\Pi$, logical operators $\neg$ (negation), $\wedge$ (conjunction) and temporal operators $\bigcirc$ (next), $\mathbf{U}$ (until). An LTL formula $\varphi$ is formed by connecting a finite set of atomic propositions with these operators. In this paper, we use a stutter-invariant fragment of LTL (denoted by $\text{LTL}_{\setminus \bigcirc}$), which excludes operation $\bigcirc$. The synthex of $\text{LTL}_{\setminus \bigcirc}$ can be found in \cite{ClarkeBook00}. We also assume that all $\text{LTL}_{\setminus \bigcirc}$ formulas have been transformed into negation normal form \cite[p.~132]{ClarkeBook00}, by adding the operator $\mathbf{R}$ (release) and replacing any negations of atomic propositions with new atomic propositions.

\emph{$\text{LTL}_{\setminus \bigcirc}$ semantics for continuous trajectories}: Let $\xi$ be a continuous-time trajectory defined on $ \mathbb{R}_{\geq 0}$ and $\varphi$ be a $\text{LTL}_{\setminus \bigcirc}$ formula. Let $\xi[t]$ denote the state at time $t$, and $\xi[t,\infty)$ denotes the part of the trajectory in $[t,\infty), t \geq 0$. Then the semantics of $\xi$ satisfying $\phi$, denoted by $\xi \models \varphi$, is defined as follows:
\begin{itemize}
\item $\xi \models \pi, \pi \in \Pi$, iff $\pi \in L(\xi[t_0])$;
\item $\xi \models \varphi_1 \wedge \varphi_2$ iff $\xi \models \varphi_1$ and $\xi \models \varphi_2$;
\item $\xi \models \varphi_1 \vee \varphi_2$ iff $\xi \models \varphi_1$ or $\xi \models \varphi_2$;
\item $\xi \models \varphi_1 \mathbf{U} \varphi_2$ iff there exists $t'> 0$ such that $\xi[t',\infty) \models \varphi_2$ and $\xi[t'',\infty) \models \varphi_1$ for all $t'' \in [0,t')$;
\item $\xi \models \varphi_1 \mathbf{R} \varphi_2$ iff for all $t'> 0$ either $\xi[t',\infty) \models \varphi_2$ or there exists $t''\in [0,t')$ such that $\xi[t'',\infty) \models \varphi_1$.
\end{itemize}

Assume the system state $x_k$ is measured at time $t_k$ with $t_0=0, 0 \leq t_k < t_{k+1}, k \in \mathbb{N}$. A \textit{continuous control strategy} is defined as a function $\sigma:x_0,\cdots,x_i \to \mathbf{u}_i$ that generates a control signal $\mathbf{u}_i \in \mathcal{U}$ for the horizon $[t_i,t_{i+1})$ according to the history of states $x_0,\cdots,x_i$.

We are now ready to formulate the main control synthesis problem this paper aims to address.

\textbf{\textit{Continuous Synthesis Problem}}: Given a continuous-time control system $\T$ and an $\text{LTL}_{\setminus \bigcirc}$ specification $\varphi$, find a nonempty set of initial states $X_0$ and a control strategy $\sigma$ such that the resulting solutions of $\T$ satisfy $\varphi$. The specification $\varphi$ is said to be \emph{realizable} for $\T$ if such $X_0$ exists.

\section{Finite Abstractions with Robustness Margins}

This section is devoted to formally defining a notion of abstractions useful for solving robust control synthesis problems and proving their correctness and robustness guarantees when solving the continuous synthesis problem by discrete synthesis using these abstractions.

\subsection{Finite abstractions with robustness margins}

In \cite{LiuO14}, the authors introduced a notion of finite abstractions with additional robustness margins that can effectively handle a range of robustness related issues in control synthesis, including modelling uncertainty, measurement errors, and jitter or delays in control signals.

This paper aims to improve its computational procedure in two aspects. First, we define the finite abstractions with a varying (state-dependent) robustness margins while \cite{LiuO14} use fixed margins which are often conservatively chosen to cope with the worst case. Second, we construct transitions by way of local reachable set computation while the results in \cite{LiuO14} rely on a global analytical bound that can lead to spurious transitions being added due to variation in local dynamics.

To this end, we shall formally define the notion of finite abstractions with robustness margins using reachable set.

\begin{definition} \label{def:reach}
Given a control signal $\mathbf{u}\in\U$ of duration $\tau$ and a set of initial states $X_0$, the \textit{reachable set} for system (\ref{eq:1}) at time $\tau$ under this control signal $\mathbf{u}$ is defined by
\begin{equation*}
\mathcal{R}_{\mathbf{u},\,X_0}(\tau):=\{ \xi(\tau,\,x_0,\,\mathbf{u})|\, x_0 \in X_0\}.
\end{equation*}
The \textit{reachable tube} for system (\ref{eq:1}) over the interval $[0,\,\tau]$ is the union of all reachable sets during this time interval, which is
\begin{equation*}
\mathcal{R}_{\mathbf{u},\,X_0}([0,\,\tau]):=\bigcup_{t \in [0,\,\tau]}\{ \xi(t,\,x_0,\,\mathbf{u})|\, x_0 \in X_0\}.
\end{equation*}
With a fixed $u\in U$ and $\tau\in\Real_{>0}$, $\mathcal{R}_{u,\,X_0}(\tau)$ and $\mathcal{R}_{u,\,X_0}([0,\,\tau])$ are interpreted as $u$ being a constant control signal on $[0,\tau]$.
\end{definition}

We are now ready to define finite abstractions with robustness margins using reachable set.

\begin{definition} \label{def:abs}

Given $\delta \in \mathbb{R}^n_{>0}$ and functions $\Gamma_i:\,X \to \mathbb{R}^n_{\geq 0}, i=1,2$, a finite transition system
$$\hat{\mathcal{T}}:=(\hat{\mathcal{Q}},\,\hat{\mathcal{Q}}_0,\,\hat{\mathcal{A}},\,\rightarrow_{\hat{\mathcal{T}}},\,\hat{\Pi},\,\hat{L})$$
is said to be a $(\Gamma_1,\Gamma_2,\delta)$-\textit{abstraction} of the continuous-time control system $\T=(X,X_0,U,f,\Pi,L)$, denoted by $\mathcal{T} \preceq_{(\Gamma_1,\Gamma_2,\delta)} \hat{\mathcal{T}}$, if there exists an abstraction map $\Omega:\, X \to \mathcal{\hat{Q}}$ such that
\begin{itemize}
\item $\hat{\mathcal{Q}}$ is a finite subset of $X$;
\item $\hat{\mathcal{Q}}_0=\bigcup_{x\in X_0}\set{\Omega(x)}$;
\item $\hat{\mathcal{A}}$ is a finite subset of $\mathcal{U}$;
\item $(\hat{q},\,\hat{\mathbf{u}},\,\hat{q}') \in \rightarrow_{\hat{\mathcal{T}}}$ if, under $\hat{\mathbf{u}} \in \hat{\mathcal{A}}$ with duration $\tau$, $\hat{q}$ and $\hat{q}'$ satisfy
$$\big(\Omega^{-1}(\hat{q}') \oplus \mathcal{B}_{\Gamma_2(\hat{q}')}\big) \cap \mathcal{R}_{\mathbf{\hat{u}},\,\Omega^{-1}(\hat{q}) \oplus \mathcal{B}_{\Gamma_1(\hat{q})}}(\tau) \neq \varnothing;$$
\item $\hat{L}:\hat{\mathcal{Q}} \to 2^{\hat{\Pi}}$ is defined by $\hat{L}(\hat{q})=\cap_{x \in \mathcal{B}_{\delta}(\hat{q}) \cap X} L(x)$, $\hat{\Pi}=\Pi$.
\end{itemize}
\end{definition}

The parameter $\delta$ is used to guarantee that specifications are satisfied  even if the controller is synthesized using a finite abstraction with approximation errors. The functions $\Gamma_{1,2}$ provide additional robustness margins that varies with respect to local dynamics to account for imperfections such as system delay, measurement or modelling errors, at the price of increasing the nondeterminism in the abstraction.

\begin{example}
A common and practical type of imperfections involves delays and measurement errors (e.g., noise or quantization). Consider the system $\T$ with a continuous control strategy $\sigma$ subjects to a measurement delay $h(t)\in [0,\Delta]$, $\Delta \in \mathbb{R}_{\geq 0}$, and an error $e(t)$ with $|e(t)|\leq \varepsilon \in \mathbb{R}^n_{\geq 0}$, the system dynamics becomes
\begin{equation} \label{eq:1b}
\left\{
\begin{aligned}
\dot{x}(t)&=f(x(t),\,\mathbf{u}_i),\, \mathbf{u}_i=\sigma(\hat{x}(t_0),\cdots,\hat{x}(t_i)),\\
\hat{x}(t_i)&=x(t_i-h(t_i))+e(t_i).
\end{aligned}
\right.
\end{equation}
where $\hat{x}$ denotes the measurement of system states, $t \in [t_i,t_{i+1}), t_0=0, t_i<t_{i+1}, i \in \mathbb{N}$ and $\tau_i=t_{i+1}-t_i$ is the time duration of $\mathbf{u}_i$.

\end{example}

\subsection{Discrete synthesis problem}

An $\text{LTL}_{\setminus \bigcirc}$ formula can be interpreted over paths of $\hat{\mathcal{T}}$. A \emph{path} of $\hat{\mathcal{T}}$ is a sequence of states $\hat{\rho}=\hat{q}_0\hat{q}_1\hat{q}_2\cdots$ under the the corresponding action $ \hat{a}_i \in \hat{\mathcal{A}}$ at each state $\hat{q}_i \in \hat{\mathcal{Q}}$ while satisfying $(\hat{q}_i,\,\hat{a}_i,\,\hat{q}_{i+1}) \in \rightarrow_{\hat{\mathcal{T}}},\, i \in \mathbb{N}$.

\emph{$\text{LTL}_{\setminus \bigcirc}$ semantics for discrete sequences}: Let $\rho=q_0q_1q_2\cdots$ be an infinite discrete sequence and $\varphi$ be an $\text{LTL}_{\setminus \bigcirc}$ formula. Let $\rho[i,\infty)$ denote the subsequence $q_iq_{i+1}\cdots, i \in \mathbb{N}$. Then semantics of $\rho$ satisfying $\varphi$, denoted by $\rho \models \varphi$, is defined as follows:
\begin{itemize}
\item $\rho \models \pi$, $\pi \in \Pi$, iff $\pi \in L(q_0)$;
\item $\rho \models \varphi_1 \wedge \varphi_2$ iff $\rho \models \varphi_1$ and $\rho \models \varphi_2$;
\item $\rho \models \varphi_1 \vee \varphi_2$ iff $\rho \models \varphi_1$ or $\rho \models \varphi_2$;
\item $\rho \models \varphi_1 \mathbf{U} \varphi_2$ iff there exists $j\geq 0$ such that $\rho[j,\infty) \models \varphi_2$ and $\rho[k,\infty) \models \varphi_1$ for all $0 \leq k < j$;
\item $\rho \models \varphi_1 \mathbf{R} \varphi_2$ iff for all $j\geq 0$ either $\rho[j,\infty) \models \varphi_2$ or there exists some $0\le k<j$ such that $\rho[k,\infty) \models \varphi_1$.
\end{itemize}

Similar to continuous control strategy, a \textit{discrete control strategy} for $\hat{\mathcal{T}}$ is a function $\hat{\sigma}:\hat{q}_0,\cdots,\hat{q}_i \to \hat{a}_i$ that maps the history path to a control action. Then we formulate the discrete synthesis problem as follows.

\textbf{\textit{Discrete Synthesis Problem}} Given a finite transition system $\hat{\mathcal{T}}$ and an $\text{LTL}_{\setminus \bigcirc}$ specification $\varphi$, find a nonempty set of initial states $\hat{X}_0$ and a control strategy $\hat{\sigma}$ such that any resulting path satisfies $\varphi$. If such $\hat{X}_0$ exists, then $\varphi$ is said to be realizable for $\hat{\mathcal{T}}$.

\subsection{Correctness and robustness guarantees}

In general, the existence of a discrete control strategy $\hat{\sigma}$ that solves the discrete synthesis problem with an $\text{LTL}_{\setminus \bigcirc}$ specification $\varphi$ does not guarantee that a control strategy exists for the continuous synthesis problem with the same specification.

As indicated in Definition \ref{def:abs}, $\hat{\mathcal{T}}$ requires the same propositions of $\T$ to hold within a neighbourhood of radius $\delta$, which is more restrictive. This is because the discrete strategy only guarantees that a sequence of sampled states satisfy a given specification and the parameter $\delta$ accounts for the possible mismatches of the inter-sample states. In addition, the robustness margin functions $\Gamma_i$ ($i=1,2$) are chosen to account for possible imperfections.

To formally reason about the correctness and robustness guarantees of solving the continuous synthesis problem by discrete synthesis using finite abstractions with robustness margins, the following theorem gives a sufficient condition for the realizability of the continuous synthesis problem by the realizability of the discrete synthesis problem.

\begin{thm} \label{thm1}
Given a continuous-time control system $\mathcal{T}$, its $(\Gamma_1,\Gamma_2,\delta)$-abstraction $\hat{\mathcal{T}}$, and an $\text{LTL}_{\setminus \bigcirc}$ formula $\varphi$,
\begin{enumerate}
\item[(i)] (correctness) $\varphi$ being realizable for $\hat{\mathcal{T}}$ implies that  $\varphi$ is realizable for $\mathcal{T}$, provided that, for all $(\hat{q},\,\hat{\mathbf{u}},\,\hat{q}') \in \rightarrow_{\hat{\mathcal{T}}}$,
\begin{equation} \label{eq:thm1}
\mathcal{R}_{\hat{\mathbf{u}},\Omega^{-1}(\hat{q}) \oplus \mathcal{B}_{\Gamma_1(\hat{q})}}(\text{dom}(\hat{\mathbf{u}}))\subseteq \mathcal{B}_\delta(\hat{q}).
\end{equation}
In particular, if $\hat{\mathcal{T}}$ satisfies $\varphi$ with $\hat{\sigma}$ and $\hat{\mathcal{Q}}_0$, then $\varphi$ is realizable for $\mathcal{T}$ using $X_0=\cup_{q\in \hat{\mathcal{Q}}_0}\Omega^{-1}(q)$ and $\sigma(x_0,\cdots,x_i)=\hat{\sigma}(\Omega(x_0),\cdots,\Omega(x_i))$ where $x_0,\cdots,x_i$ is the sequence of measured states.
\item[(ii)] (robustness) if the system is subjected to measurement delays and errors defined in (\ref{eq:1b}), then the same statement holds true, provided additionally that the robustness margins $\Gamma_{i}$ ($i=1,2$) satisfy that, for all $\hat{\mathbf{v}} \in \hat{\mathcal{A}}$ and $\hat{q}\in\hat{\Q}$, $\Gamma_2(\hat{q}) \geq \varepsilon$ and
\begin{equation} \label{eq:thm2}
\mathcal{R}_{\hat{\mathbf{v}},\Omega^{-1}(\hat{q}) \oplus \mathcal{B}_\varepsilon}([0,\Delta]) \subseteq \Omega^{-1}(\hat{q}) \oplus \mathcal{B}_{\Gamma_1(\hat{q})}.
\end{equation}
\end{enumerate}
\end{thm}

\begin{proof}
(i) The realizability of $\varphi$ for $\hat{\T}$ implies that there exists an initial set $\hat{\mathcal{Q}}_0$ and a discrete control strategy $\hat{\sigma}$ for $\hat{\mathcal{T}}$ such that all the possible controlled paths from any initial state in $\hat{\mathcal{Q}}_0$ satisfies $\varphi$ (note that $\hat{\mathcal{T}}$ is nondeterministic). We need to show the realizability of $\varphi$ for $\T$. For this purpose, we define an initial set $X_0=\cup_{q\in \hat{\mathcal{Q}}_0}\Omega^{-1}(q)$ and a continuous control strategy by
$$\sigma(x_0,\cdots,x_i)=\hat{\u}_i=\hat{\sigma}(\Omega(x_0),\cdots,\Omega(x_i)),$$
where $x_0,\cdots,x_i$ is a sequence of measured states. We write $\hat{q}_i=\Omega(x_i)$ for all $i\ge 0$ and apparently $\hat{q}_0 \in \hat{\mathcal{Q}}_0$. In addition, we denote by $\tau_i$ the duration of $\hat{\u}_i$ and let $t_0=0,t_i=\sum_{k=0}^{i-1}\tau_k,i=1,2,\cdots$. Denote by $\xi$ the trajectory of $\T$ starting from $x_0$ under the control strategy $\sigma$ and by $\hat{\rho}$ the path $\hat{q}_0\hat{q}_1\hat{q}_2\cdots$. This correspondence is illustrated by the diagram below:
\begin{center}
\includegraphics[scale=0.5]{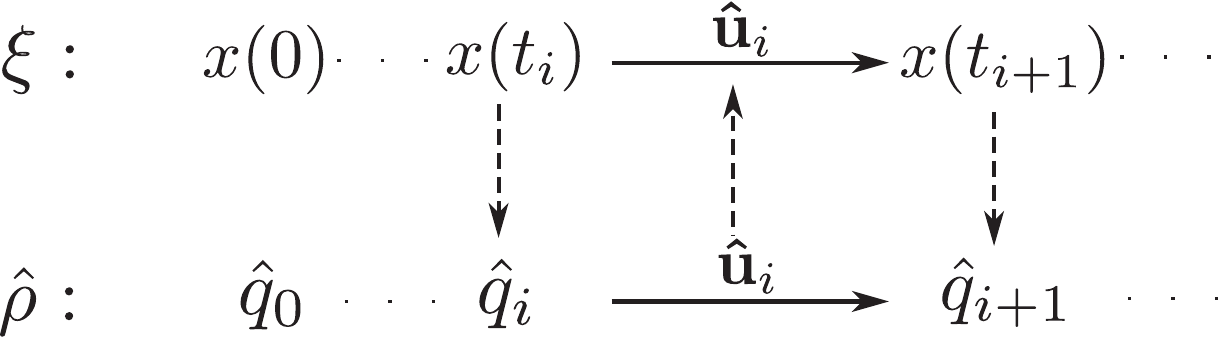}
\end{center}

The proof consists of two steps: (A) to show that the path $\hat{\rho}=\hat{q}_0\hat{q}_1\hat{q}_2\cdots$ is a valid path in $\hat{\T}$ and, as a result, $\hat{\rho}\models \varphi$; (B) to show from $\hat{\rho}\models \varphi$ that $\xi \models \varphi$.

To show (A), note that, since $x(t_i) \in \Omega^{-1}(\hat{q}_i)$ for all $i\ge 0$, we have
$$x(t_{i+1}) \in \mathcal{R}_{\hat{\mathbf{u}}_i,\Omega^{-1}(\hat{q}_i) \oplus \mathcal{B}_{\Gamma_1(\hat{q}_i)}}(\tau_i).$$
It follows from the definition of the transitions of $\hat{\mathcal{T}}$ that $(\hat{q}_{i},\hat{u}_i,\hat{q}_{i+1})\in\rightarrow_{\hat{\mathcal{T}}}$ for all $i\ge 0$.

To show (B), we prove $\xi \models \varphi$ from $\hat{\rho}\models \varphi$ by induction on the form of $\text{LTL}_{\setminus \bigcirc}$ formulas. In fact, we will prove a stronger statement: for each $k\ge 0$, $\hat{\rho}[k,\infty)$ implies that $\xi[t,\infty) \models \varphi_1$ for all $t\in [t_k,t_{k+1})$.

For $\varphi=\pi \in \Pi$, $\hat{\rho}[k,\infty) \models \pi$ iff $\pi \in \hat{L}(\hat{q}_k)$. Since
$$x(t)\in \mathcal{R}_{\hat{\mathbf{u}}_k,\Omega^{-1}(\hat{q}_k) \oplus \mathcal{B}_{\Gamma_1(\hat{q}_k)}}([0,\tau_k]),\quad\forall t\in[t_k,t_{k+1}),$$
we have $\pi \in \hat{L}(\hat{q}_k) \subseteq L(x(t))$, i.e., $\xi[t,\infty) \models \varphi=\pi$, for all $t\in [t_k,t_{k+1})$.

The cases for $\xi \models \varphi$ when $\varphi=\varphi_1 \wedge \varphi_2$ or $\varphi=\varphi_1 \vee \varphi_2$ are straightforward to prove. We focus on the case $\varphi=\varphi_1 \mathbf{U} \varphi_2$. Assume $\hat{\rho}[k,\infty) \models \varphi$, which means that there exists some $j\geq k$ such that $\hat{\rho}_0[j,\infty) \models \varphi_2$ and $\hat{\rho}_0[i,\infty) \models \varphi_1$ for all $i$ such that $k \leq i < j$. By the inductive assumption, we have $\xi[t,\infty) \models \varphi_2$ for all $t\in [t_j,t_{j+1})$ and $\xi[t,\infty) \models \varphi_1$ for all $t\in [t_i,t_{i+1})$ and all $i$ such that $k \leq i < j$. This indeed implies that $\xi[t,\infty) \models \varphi=\varphi_1 \mathbf{U} \varphi_2$, for all $t\in [t_k,t_{k+1})$. The proof for the case $\varphi=\varphi_1 \mathbf{R} \varphi_2$ is similar and therefore omitted.

(ii) Now consider system (\ref{eq:1b}) for robustness. The key difference now is that measured states are delayed versions of the longer true states affected by noise. Denote by $\hat{x}(t_i) \in \mathcal{B}_\varepsilon(x(t_i))$ the measured value of $x(t_i)$ and let $\hat{q}_i=\Omega(\hat{x}(t_i))$ for all $i\ge 0$. The corresponding continuous control strategy becomes
$$\sigma(\hat{x}(t_0),\cdots,\hat{x}(t_i))=\hat{\u}_i=\hat{\sigma}(\hat{q}_0,\cdots,\hat{q}_i).$$
Each control action $\hat{\u}_i$ is activated when the true state moves to $x(t_i)'=x(t_i+h(t_i)).$ The correspondence between the evolution of a true trajectory and the sequence of measure states are illustrated in the following diagram:
\begin{center}
      \includegraphics[scale=0.5]{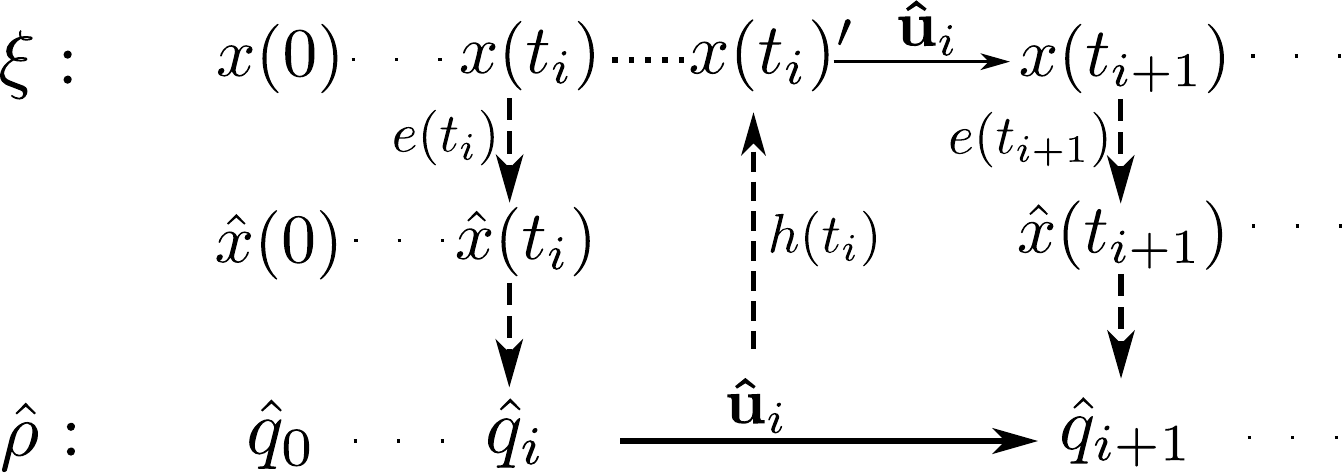}
\end{center}
We still need to show the two steps (A) and (B) as in part (i). We start with (A), i.e., show that the path $\hat{\rho}=\hat{q}_0\hat{q}_1\hat{q}_2\cdots$ is a valid path in $\hat{\T}$. Note that, according to (\ref{eq:thm2}), we have
$$
x(t_i)' \in \mathcal{R}_{u_{i-1},\Omega^{-1}(\hat{q}_i) \oplus \mathcal{B}_\varepsilon}([0,\Delta]) \subseteq \Omega^{-1}(\hat{q}_i) \oplus \mathcal{B}_{\Gamma_1(\hat{q}_i)}.
$$
Therefore
$$
x(t_{i+1}) \in \mathcal{R}_{\hat{\mathbf{u}}_i,\Omega^{-1}(\hat{q}_i) \oplus \mathcal{B}_{\Gamma_1(\hat{q}_i)}}(\tau_i) \subseteq \mathcal{B}_\delta(\hat{q}_{i}).
$$
Since $\hat{x}(t_{i+1}) \in \mathcal{B}_\varepsilon(x(t_{i+1}))$ and $\hat{q}_{i+1}=\Omega(\hat{x}(t_{i+1}))$, we have $x(t_{i+1})\in \Omega^{-1}(\hat{q}_{i+1})\oplus\mathcal{B}_\varepsilon$. Considering that the transitions for $\hat{\T}$ are constructed according to Definition \ref{def:abs} with $\Gamma_2 \geq \varepsilon$, the transition $(\hat{q}_i,\,\hat{\mathbf{u}}_i,\,\hat{q}_{i+1})$ is indeed included in $\rightarrow_{\hat{\mathcal{T}}}$.

Proving step (B) by induction is similar to that for part (i). We prove the claim: for each $k\ge 0$, $\hat{\rho}[k,\infty)$ implies that $\xi[t,\infty) \models \varphi_1$ for all $t\in [t_k,t_{k+1})$. Note that we have $t_k+h(t_k)\in [t_k,t_{k+1})$ and $t_{k+1}-t_k-h(t_k)=\tau_k$, the duration of $\hat{u}_k$. We only prove the case for atomic propositions and the rest is similar to that for part (i).

For $\varphi=\pi \in \Pi$, $\hat{\rho}[k,\infty) \models \pi$ iff $\pi \in \hat{L}(\hat{q}_k)$. Note first that, by (\ref{eq:thm2}),
$$x(t)\in \mathcal{R}_{\hat{\mathbf{u}}_{k-1},\Omega^{-1}(\hat{q}_k)\oplus\mathcal{B}_{\varepsilon}}([0,\Delta])\subseteq \Omega^{-1}(\hat{q}_i) \oplus \mathcal{B}_{\Gamma_1(\hat{q}_k)}\subseteq \mathcal{B}_\delta(\hat{q}_{k})
$$
for all $t\in[t_k,t_{k}+h(t_k)]$. This and (\ref{eq:thm1}) further imply that
$$
x(t)\in \mathcal{R}_{\hat{\mathbf{u}}_{k},\Omega^{-1}(\hat{q}_k) \oplus \mathcal{B}_{\Gamma_1(\hat{q}_i)}}([0,\tau_k])\subseteq \mathcal{B}_\delta(\hat{q}_{k})
$$
for all $t\in[t_k+h(t_k),t_{k+1})$. Consequently, we have $\pi \in \hat{L}(\hat{q}_k) \subseteq L(x(t))$, i.e., $\xi[t,\infty) \models \varphi=\pi$, for all $t\in [t_k,t_{k+1})$. \qed
\end{proof}

\section{Reachable Set Over-approximation Based on Linearization and Error Estimation}

A key step in constructing finite abstractions with robustness margins defined in the previous section is to compute the reachable sets for nonlinear systems. In practice, exact reachable sets of nonlinear systems are difficult to obtain and thus their approximations are usually computed.
For example, reachable set over-approximation is implicitly required by the abstraction procedures in \cite{PolaGT08,ZamaniPMT12,LiuO14}, where analytical bounds, usually obtained by Lyapunov-like functions, are used to roughly estimate the evolution of trajectories. A more precise computation of reachable sets has the potential to significantly reduce the spurious transitions in the abstraction.

In this section, we present a linearization-based method for the computation of reachable sets for nonlinear systems. For simplicity, we only consider constant control signals, which suffice for the computation of finite abstractions by discretization-based methods to be discussed in Section \ref{sec:discretization}.

\subsection{Reachable set computation for linear systems}

Consider a class of affine control systems of the form
\begin{equation}
\dot{x}(t)=Ax(t)+b+u(t) \label{eq:linaff}
\end{equation}
where $b \in \mathbb{R}^n$ is a constant vector, $x(t) \in X$ is the state, $u(t) \in U$ is the control signal, and $U \subseteq \mathbb{R}^m$ is a compact convex set.

Similar to Definition \ref{def:reach}, given an initial set of states $X_0 \subseteq X$, we denote by $\mathcal{R}_{X_0}^{L}(\tau)$ the set of states that are reachable at time $\tau \in \mathbb{R}_{\geq 0}$ under $U$, which is defined by
\begin{equation*}
\begin{split}
\mathcal{R}^L_{X_0}(\tau):=\{ x(\tau) \in X|\,&\dot{x}(t)=Ax(t)+b+u(t),\forall t \in [0,\tau],\\
&u(t) \in U, x(0) \in X_0\}.
\end{split}
\end{equation*}
The reachable tube over the interval $[0,\,\tau]$ is defined by
\begin{equation*}
\mathcal{R}^L_{X_0}([0,\tau]):=\bigcup_{t \in [0,\tau]}\mathcal{R}^L_{X_0}(t).
\end{equation*}

Since the control input $u(t)$ is chosen arbitrarily from the set $U$, both the reachable set and tube are difficult to be computed exactly. For linear control systems, their convex over-approximations are used instead (see, e.g., Lemmas 1 and 2 in \cite{GuernicG10}). The convex hull of two convex sets, which is defined by
$$
\text{CH}(\mathcal{X},\mathcal{Y})=\{\lambda x+(1-\lambda)y|\,x \in \mathcal{X}, y \in \mathcal{Y}, \lambda \in [0,1] \},
$$
is used to compute the reachable tube. For the linear affine control systems, we give the following proposition to over-approximate the reachable sets and tubes.

\begin{prop} \label{prop1}
For a linear affine control system (\ref{eq:linaff}), given a compact convex set $X_0 \subseteq X$ and a time $\tau \in \mathbb{R}_{\geq 0}$, let
\begin{equation} \label{eq:4a}
\begin{split}
Y(\tau)&=e^{A\tau}X_0 \oplus \{G(A,\tau)b \} \oplus \tau U \oplus \mathcal{B}_{\beta_\tau},\\
Y([0,\tau])&=\text{CH}(X_0, Y(\tau) \oplus \mathcal{B}_{\alpha_\tau+\gamma_\tau}),
\end{split}
\end{equation}
where
\begin{equation} \label{eq:4b}
\begin{split}
\alpha_\tau&=(e^{\tau \|A\|}-1-\tau \|A\|)\max_{x \in X_0}\|x\| \mathbf{1},\\
\beta_\tau&=(e^{\tau \|A\|}-1-\tau \|A\|)\|A\|^{-1}\max_{u \in U}\|u\| \mathbf{1},\\
\gamma_\tau&=(e^{\tau \|A\|}-1-\tau \|A\|)\|A\|^{-1}\|b\|\mathbf{1},
\end{split}
\end{equation}
with $\norm{\cdot}$ as the infinity norm, $\mathbf{1} \in \mathbb{R}^n$ representing the vector of ones, i.e., each element of it equals to 1, and $G(A,\tau):=\int_0^{\tau}e^{A(\tau-t)}dt$. Then
\begin{equation*}
\begin{split}
\mathcal{R}^L_{X_0}(\tau) &\subseteq Y(\tau),\\
\mathcal{R}^L_{X_0}([0,\tau]) &\subseteq Y([0,\tau]).
\end{split}
\end{equation*}
\end{prop}

\begin{proof}
Denote by $x(t), t \in [0,\tau]$, a trajectory of the system from a initial state $x_0 \in X_0$ under an input $u(t) \in U$, and
\begin{equation*}
\begin{split}
x(t)=&e^{tA}x_0+\int_0^te^{A(t-s)}b\,ds\\
&+\int_0^{t}u(s)ds+\int_0^{t}(e^{A(t-s)}-I)u(s)ds\\
=&e^{tA}x_0+G(A,t)b+tu^*(t)\\
&+\int_0^{t}(e^{A(t-s)}-I)u(s)ds,
\end{split}
\end{equation*}
where $u^*(t)=\frac{1}{t}\int_0^tu(s)ds \in U$ for that $U$ is convex. We estimate $x(t)$ by $\hat{x}(t)$, which is given by
\begin{equation*}
\hat{x}(t)=x_0+\frac{t}{\tau}(e^{\tau A}-I)x_0+\frac{t}{\tau}G(A,\tau)b+tu^*(t).
\end{equation*}
Then
\begin{equation} \label{eq:5}
\begin{split}
\|x(t)-\hat{x}(t)\| \leq &\|e^{tA}x_0-x_0-\frac{t}{\tau}(e^{\tau A}-I)x_0\|\\
& + \|G(A,t)b-\frac{t}{\tau}G(A,\tau)b\| \\
&+ \|\int_0^{t}(e^{A(t-s)}-I)u(s)ds\|\\
\leq &\frac{t}{\tau}(\alpha_\tau+\gamma_\tau+\beta_\tau).
\end{split}
\end{equation}
This means there exists a vector $\tilde{x}(t)$ in $\mathcal{B}_{\alpha_\tau+\gamma_\tau+\beta_\tau}$ such that
\begin{equation*}
\begin{split}
x(t)&=\hat{x}(t)+\frac{t}{\tau}\tilde{x}(t)\\
&=(1-\frac{t}{\tau})x_0+\frac{t}{\tau}(e^{\tau A}+G(A,\tau)b+tu^*(t)+\tilde{x}(t)).
\end{split}
\end{equation*}
Therefore
\begin{align*}
\mathcal{R}^L_{X_0}([0,\tau]) & \subseteq \text{CH}(X_0, e^{A\tau}X_0 \oplus \{G(A,\tau)b \} \oplus \tau U \oplus \mathcal{B}_{\alpha_\tau+\gamma_\tau+\beta_\tau})\\
&=Y([0,\tau]).
\end{align*}

The state estimation error at time $\tau$ reduces to $\|x(\tau)-\hat{x}(\tau)\|\leq \beta_\tau$ by setting $t=\tau$ in (\ref{eq:5}). Thus $\mathcal{R}^L_{X_0}(\tau) \subseteq e^{A\tau}X_0 \oplus \{G(A,\tau)b \} \oplus \tau U \oplus \mathcal{B}_{\beta_\tau}=Y(\tau)$. \qed
\end{proof}

\begin{remark}
Proposition \ref{prop1} differs from \cite{GuernicG10} in considering affine systems. Defining $v(t):=b+u(t), v(t) \in V=\{b\} \oplus U$, the method in \cite{GuernicG10} can also be applied. Yet when $u(t)$ is small compared to $b$, the size of $Y(\tau)$ computed by proposition \ref{prop1} is smaller because of a smaller bloating parameter $\beta_\tau$.

\end{remark}

\subsection{Reachable set computation for nonlinear systems}

Reachable set over-approximation for nonlinear systems obtained by a global analytical function can be conservative. To obtain a relatively tighter over-approximation of the one-step reachable set of nonlinear systems, we can write the nonlinear system dynamics as the sum of its linearization in a local area and an approximation error term.

More specifically, for a nonlinear system (\ref{eq:1}) under a constant control input $u\in \mathcal{U}$, the dynamics around a center point $x^*\in X$ can be approximated by its first-order Taylor expansion with a Lagrangian remainder:
\begin{equation} \label{eq:lin}
\dot{x}(t) = A_{x^*}(x(t)-x^*)+f(x^*,u)+d_{x^*}(t),
\end{equation}
where $A_{x^*}=\partial f/ \partial x|_{x^*}$, and $d_{x^*}(t)=(d_1(t),\cdots,d_n(t)) \in \mathbb{R}^n$ is the approximation error with
$$
d_i(t)=\frac{1}{2}(x(t)-x^*)^TH_i(z_i(t))(x(t)-x^*),
$$
$$
H_i(z_i(t))=\frac{\partial^2 f_i}{\partial x^2}\bigg|_{z_i(t)},
$$
and $z_i(t) \in \mathcal{B}_{|x(t)-x^*|}(x^*)$.

If the system trajectory does not exceed a predefined linearization area $\mathcal{B}_r(x^*)$, where $r \in \mathbb{R}^n_{>0}$, then $d_{x^*}(t)$ belongs to a convex set $\mathcal{D}_{x^*}(r)$ given by
\begin{equation} \label{eq:d}
\begin{split}
\mathcal{D}_{x^*}(r)=\{d=(d_1,\,\dots,\,d_n)|\,&d_i=\frac{1}{2}x^TH_i(z_i)x,\\
& x \in \mathcal{B}_r,\, z_i \in \mathcal{B}_r(x^*)\}.
\end{split}
\end{equation}

Defining $\tilde{x}(t):=x(t)-x^*$, (\ref{eq:lin}) is in the form of (\ref{eq:linaff}). Thus, the reachable set and tube of the nonlinear control system (\ref{eq:1}) can be computed using Proposition \ref{prop1} locally.

\subsection{Reachable set computation using zonotopes}
Since set operations, such as linear transformation, addition and multiplication, are used extensively in the computation of reachable sets, a proper set representation can help expedite the computational process. To this end, zonotope representation is attractive for its efficiency in the aforementioned set operations (see, e.g., \cite{Girard05,GirardGM06,Althoff14}).

\begin{definition} \label{def:zonotope}
A \textit{zonotope} is a set represented as
\begin{equation*}
\mathcal{Z}:=\left\{ x \in \mathbb{R}^n|\,x=c+\sum_{i=1}^l\lambda_ig^{(i)},\,\lambda_i \in[-1,\,1]\right\},
\end{equation*}
where $c,\,g^{(i)}(i=1,\,2,\,\dots,\,l) \in \mathbb{R}^n$ are called the central vector and generators, respectively; $l$ is the number of generators. It is often denoted as $\mathcal{Z}=(c,\,g^{(1)},\,\dots,\,g^{(l)})$.
\end{definition}

The addition of two zonotopes $\mathcal{Z}_1=(c_1,\,g_1^{(1)},\,\dots,\,g_1^{(l_1)})$ and $\mathcal{Z}_2=(c_2,\,g_2^{(1)},\,\dots,\,g_2^{(l_2)})$ and the multiplication of a zonotope with a matrix $M \in \mathbb{R}^{n\times n}$ can be easily derived as
\begin{equation*}
\begin{split}
\mathcal{Z}_1 \oplus \mathcal{Z}_2 &= (c_1+c_2,\,g_1^{(1)},\,\dots,\,g_1^{(l_1)},\,g_2^{(1)},\,\dots,\,g_2^{(l_2)}),\\
M\mathcal{Z}_1 &=(Mc_1,\,Mg_1^{(1)},\,\dots,\,Mg_1^{(l_1)}).
\end{split}
\end{equation*}

For a zonotope with $l$ generators in $\mathbb{R}^n$, $l/n$ is called the \textit{order} of the zonotope.

\begin{example}
The set $\mathcal{B}_r$ with $r=(r_1,\cdots,r_n), r_i \in \mathbb{R}_{>0}$ can be written in the form of zonotope as
\begin{equation} \label{eq:br}
\mathcal{Z}_{\mathcal{B}_r}=(0,\,g_r^{(1)},g_r^{(2)},\cdots,,g_r^{(n)}),
\end{equation}
where $g_r^{(i)} \in \mathbb{R}^n$ is a vector with all the elements being zero except that the $i$th element is $r_i$, $i=1,2,\cdots,n$.
\end{example}

The approximation error $\mathcal{D}_{x^*}(r)$ as in (\ref{eq:d}) can be over-approximated using the quadratic map \cite{Althoff14}.
Instead of computing $H_i(z_i)$ for every $z_i \in \mathcal{B}_r(x^*)$, we enclose it by an interval matrix $\overline{H}_i(x^*)$. Denote by $\overline{h}_{ij}$ the element of the $i$th row and $j$th column of $\overline{H}_i(x^*)$, then $\overline{h}_{ij}=[h_{ij}^l,h_{ij}^u]$, where $h_{ij}^l$ and $h_{ij}^u$ is the minimum and maximum values of $\overline{h}_{ij}$ in the linearization area respectively. Using $\mathcal{Z}_{\mathcal{B}_r}$ defined in (\ref{eq:br}), we can compute an over-approximation of $\mathcal{D}_{x^*}(r)$ by
\begin{equation} \label{eq:dzono}
\mathcal{D}_{x^*}(r)\subseteq \overline{\mathcal{D}}_{x^*}(r):=\text{quad}(\overline{H}_i(x^*),\mathcal{Z}_{\mathcal{B}_r}),
\end{equation}
where $\text{quad}(\cdot,\cdot)$ is the quadratic map defined in \cite{Althoff14}.

The convex hull operation of two zonotopes can be over-approximated by (see \cite{Girard05,Althoff10} for more details)
\begin{equation*}
\begin{split}
\overline{\text{CH}}(\mathcal{Z}_1,\mathcal{Z}_2)=\frac{1}{2} ( &c_1+c_2, g_1^{(1)}+g_2^{(1)}, \cdots, g_1^{(l)}+g_2^{(l)},\\
&c_1-c_2, g_1^{(1)}-g_2^{(1)},\cdots, g_1^{(l)}-g_2^{(l)} ).
\end{split}
\end{equation*}

To sum up, we give the following proposition, which aims to over-approximate the local reachable sets of nonlinear systems using zonotopes.

\begin{prop}
Given a nonlinear control system $\T$, the function $\Gamma_1:X \to \mathbb{R}^n_{\geq 0}$, an abstraction map $\Omega: X \to \hat{\mathcal{Q}}$ and a finite set of constant control actions $\hat{\mathcal{A}}$, for any $\hat{q} \in \hat{\mathcal{Q}}$ and $\hat{\mathbf{u}} \in \hat{\mathcal{A}}$ with $\hat{\mathbf{u}}(t)=\hat{u} \in U, \forall t \in [0,\tau]$, denote
\begin{equation} \label{eq:x}
\begin{split}
X_{\hat{q}}=\Omega^{-1}(\hat{q}) \oplus \mathcal{B}_{\Gamma_1(\hat{q})},\; \tilde{X}_{\hat{q}}=\{-\hat{q}\} \oplus X_{\hat{q}}.
\end{split}
\end{equation}
The reachable set and tube $\mathcal{R}_{\hat{u},X_{\hat{q}}}(\tau)$ and $\mathcal{R}_{\hat{u},X_{\hat{q}}}([0,\tau])$ can be over-approximated by the sets $\overline{\mathcal{R}}_{\hat{u},X_{\hat{q}}}(\tau)$ and $\overline{\mathcal{R}}_{\hat{u},X_{\hat{q}}}([0,\tau])$, respectively, which are computed by
\begin{equation} \label{eq:reachset}
\overline{\mathcal{R}}_{\hat{u},X_{\hat{q}}}(\tau) = \set{\hat{q}} \oplus \tilde{Y}(\tau),
\end{equation}
and
\begin{equation} \label{eq:reachtube}
\overline{\mathcal{R}}_{\hat{u},X_{\hat{q}}}([0,\tau]) = \set{\hat{q}} \oplus \overline{\text{CH}}(\tilde{X}_{\hat{q}}, \tilde{Y}(\tau) \oplus \mathcal{B}_{\alpha_\tau+\gamma_\tau}),
\end{equation}
where
\begin{equation*}
\tilde{Y}(\tau)=e^{A_{\hat{q}}\tau}\tilde{X}_{\hat{q}} \oplus G(A_{\hat{q}},\tau)f(\hat{q},u) \oplus \tau \mathcal{D}_{\hat{q}}(r) \oplus \mathcal{B}_{\beta_\tau},
\end{equation*}
and $\alpha_\tau$, $\beta_\tau$, $\gamma_\tau$, $G(A_{\hat{q}},\tau)$ are defined as in Proposition \ref{prop1}.
\end{prop}

\section{Computation of Abstraction by Discretization and Zonotope Representation} \label{sec:discretization}

In this section, we discuss how to construct finite abstractions with robustness margins by grid-based discretization.

\subsection{Grid-based discretization}

Consider uniform parameters $\eta \in \mathbb{R}^n_{>0}$, $\mu \in \mathbb{R}^m_{>0}$ and a fixed sampling time $\tau_s \in \mathbb{R}_{>0}$. Let $\hat{\mathcal{Q}}=[X]_{\eta}$ be the set of states in $\hat{\T}$. In this case, $\Omega^{-1}(\hat{q}) \oplus \mathcal{B}_{\Gamma_{i}(\hat{q})}=\mathcal{B}_{\eta/2+\Gamma_{i}(\hat{q})}(\hat{q})$ ($i=1,2$). Using zonotopes with order 1, $X_{\hat{q}}$, $\tilde{X}_{\hat{q}}$ in (\ref{eq:x}) become
\begin{equation*}
\begin{split}
X_{\hat{q}}&=(\hat{q},\,g_\eta^{(1)},g_\eta^{(2)},\cdots,,g_\eta^{(n)}),\\
\tilde{X}_{\hat{q}}&=(0,\,g_\eta^{(1)},g_\eta^{(2)},\cdots,,g_\eta^{(n)}),
\end{split}
\end{equation*}
where $g_\eta^{(i)} \in \mathbb{R}^n$ is a vector with all the elements are zero except the $i$th element being $\eta/2+\Gamma_1(\hat{q})$, $i=1,2,\cdots,n$.

The set of control actions $\hat{\mathcal{A}}$ only contains the control signals that take values in $[U]_{\mu}$ and the time duration are integral multiples of $\tau_s$. Since the computation of reachable sets and tubes are only valid within the linearization area $\mathcal{B}_r(\hat{q})$, the time duration and the value of the control signals should be determined to make sure that the transitions only take place inside it. Furthermore, in order to satisfy Theorem 1, this area should belong to $\mathcal{B}_\delta(\hat{q})$; in other words, $r \leq \delta$.

\subsection{Algorithm for computing transitions}

The algorithm for computing transitions is designed to collect all the valid transitions under a grid-based discretization according to Theorem 1. The main steps are devoted to solving the key problem of determining the valid control signal duration $\tau=k\tau_s,k \in \mathbb{N}$ (if it exists) for each element in $[U]_\mu$ and state in $\mathcal{\hat{Q}}$.

Similar to a lazy control strategy, which means that the control action is kept to be the same for as long as possible, we choose $\tau=\tau_{\text{max}}$, where $\tau_{\text{max}}$ is the maximum time of a control signal under which the system remains within a predefined linearization area. A practical consideration for this is that a short time duration can potentially introduce spurious self-transitions that do not exist in the original continuous system.

Out of simplicity in implementation, we use $\hat{\tau}_{\text{max}}=p^*\tau_s, p^* \in \mathbb{N}$ as an under-approximation of $\tau_{\mathrm{max}}$, and approach it iteratively using a lower bound $a$ and an upper bound $b$ ($a,\,b\in\mathbb{N}$ and $a \leq b$).
The initial guess equals to the upper bound $b$. If the reachable set is fully inside the linearization area, which means $p^*\ge b$, the bounds shift to $[b,\,b+(b-a)]$; if the reachable set has already move outside the region, the bounds shrink to $[a,\,\lfloor \frac{a+b}{2} \rfloor]$. Considering the situation that reachable sets shrinks around the equilibriums, i.e., $\tau_{\mathrm{max}}=\infty$, we set an upper limit $N \in \mathbb{N}$ for $p$.

Algorithm 1 sketches the computation of transitions in a $(\Gamma_1,\Gamma_2,\delta)$-abstraction. For system (\ref{eq:1}), we can use constant margins satisfying $\Gamma_{1,2} \geq 0$. For system (\ref{eq:1b}), $\Gamma_2 \geq \varepsilon$ can be set as a constant, whereas the margin $\Gamma_1$ is not predefined, but chosen adaptively according to (\ref{eq:thm2}).

\begin{algorithm}
\caption{Computation of the transitions $\{ \rightarrow_{\hat{\mathcal{T}}}\}$ in a $(\Gamma_1,\Gamma_2,\delta)$-abstraction $\hat{\mathcal T}$}
\begin{algorithmic}[1]
\Require $r,\tau_s,\eta,\hat{\mathcal{A}},\hat{\mathcal{Q}}$ and $\Delta$, $\varepsilon$ ($\Delta=0, \varepsilon=0$ for (\ref{eq:1}))
\State $\Gamma_1=\varepsilon$, $\Gamma_2 \equiv \varepsilon$, $\{ \rightarrow_{\hat{\mathcal{T}}}\} \leftarrow \varnothing$
\ForAll{$\hat{q} \in \hat{\mathcal{Q}}$}
\ForAll{$\hat{u} \in \hat{\mathcal{A}}$}
\State Compute $f_{\hat{q}}$, $A_{\hat{q}}$, and $\overline{\mathcal{D}}_{\hat{q}}(r)$ by (\ref{eq:lin}) and (\ref{eq:dzono})

\State $X_0'=\varnothing$
\ForAll{$\hat{v} \in \hat{\mathcal{A}}$}
\State $X_0'=X_0' \cup \overline{\mathcal{R}}_{\hat{v},\,\mathcal{B}_{\eta/2+\varepsilon}(\hat{q})}([0,\Delta])$
\EndFor
\State Choose $\Gamma_1$ s.t. $X_0' \oplus (-\mathcal{B}_{\eta/2}(\hat{q}))\subseteq \mathcal{B}_{\Gamma_1}$
\State $X_0 = \mathcal{B}_{\eta/2+\Gamma_1}(\hat{q})$, $X_R=\varnothing$

\State $p=p_0$, $a=0$, $b=p$
\While{$(a \neq b)\, \wedge\, (p>0)\, \wedge\, (p<N)$}
\State Compute $\overline{\mathcal{R}}_{\hat{u},X_0}(p\tau_s)$, $\overline{\mathcal{R}}_{\hat{u},X_0}([0,p\tau_s])$
\If{$\overline{\mathcal{R}}_{\hat{u},X_0}([0,\,p\tau_s]) \subseteq \mathcal{B}_{r}(\hat{q})$}
\State $X_R=\overline{\mathcal{R}}_{\hat{u},X_0}(p\tau_s)$
\State $p=2b-a$, $a=b$, $b=p$
\Else
\State $p=\lfloor \frac{a+b}{2} \rfloor$, $b=p$
\EndIf
\EndWhile
\State $\tau=a\tau_s$
\If{$\mathcal{B}_{\eta/2+\Gamma_2}(\hat{q}') \cap X_R \neq \varnothing$}
\State $\{ \rightarrow_{\hat{\mathcal{T}}}\} \leftarrow (\hat{q},\,\hat{u},\,\tau,\,\hat{q}')$
\EndIf
\EndFor
\EndFor
\State \Return $\{ \rightarrow_{\hat{\mathcal{T}}}\}$
\end{algorithmic}
\end{algorithm}


\section{Comparison with Lyapunov-based Approximation}

We analyze the performance of the controllers synthesized using finite abstractions with robustness margins by two examples: the pendulum system \cite{PolaGT08}) and the automatic cruise control \cite{LiuO14}.

\subsection{Pendulum}
The pendulum model considered here is
\begin{equation*}
\begin{split}
	\begin{bmatrix}
	\dot{x}_1\\
	\dot{x}_2
	\end{bmatrix}
	&=
	\begin{bmatrix}
	x_2\\
	-\frac{g}{l}\sin x_1-\frac{k}{m}x_2+u
	\end{bmatrix},\\
	g&=9.8,\,l=5,\,m=0.5,\,k=3,\\
\end{split}
\end{equation*}
where $u \in U=[-1, 0], x \in X=[-0.5, 0]\times[-0.2, 0.2]$; $u$ is the normalized control torque; $x_1, x_2$ represent the angle (rad) and the angular rate (rad/s), respectively. The angle is measured from the perpendicular line to the current ball position. The positive direction is counter clockwise. The constants $g$, $l$, $m$, $k$ denote the gravity acceleration, rod length, mass, and friction coefficient, respectively.

The specification is given by an $\text{LTL}_{\setminus \bigcirc}$ formula $\varphi=\square \varphi_s \wedge \lozenge \square \varphi_t$ with $\varphi_s=X$ and $\varphi_t=[-0.3,\,-0.2] \times [-0.05,\,0.05]$. In our simulation, the abstraction parameters are $\tau_s=0.01s,\,r=[0.04;\,0.04],\,\eta=[0.02;\,0.02],\,\mu=0.01$. As shown in Fig.~\ref{fig:pdl} (left), the controlled system trajectory satisfies the given specification.

On the other hand, we fail to generate a controller using the abstraction based on Lyapunov-like method, as a result of its greater conservatism. We compare the number of transitions included by different reachable set computation methods. With the same partition, applying the control torque $u=-0.81$ at the state $x_1=-0.3$, $x_2=0.1$, the number of post states computed by our method is 4 while it is 49 using the Lyapunov-based method. As shown in Fig.~\ref{fig:pdl} (right), the one-step reachable set computed using our method is smaller than that using the Lyapunov-based method.

\begin{figure}[thpb]
      \centering
      \includegraphics[scale=0.4]{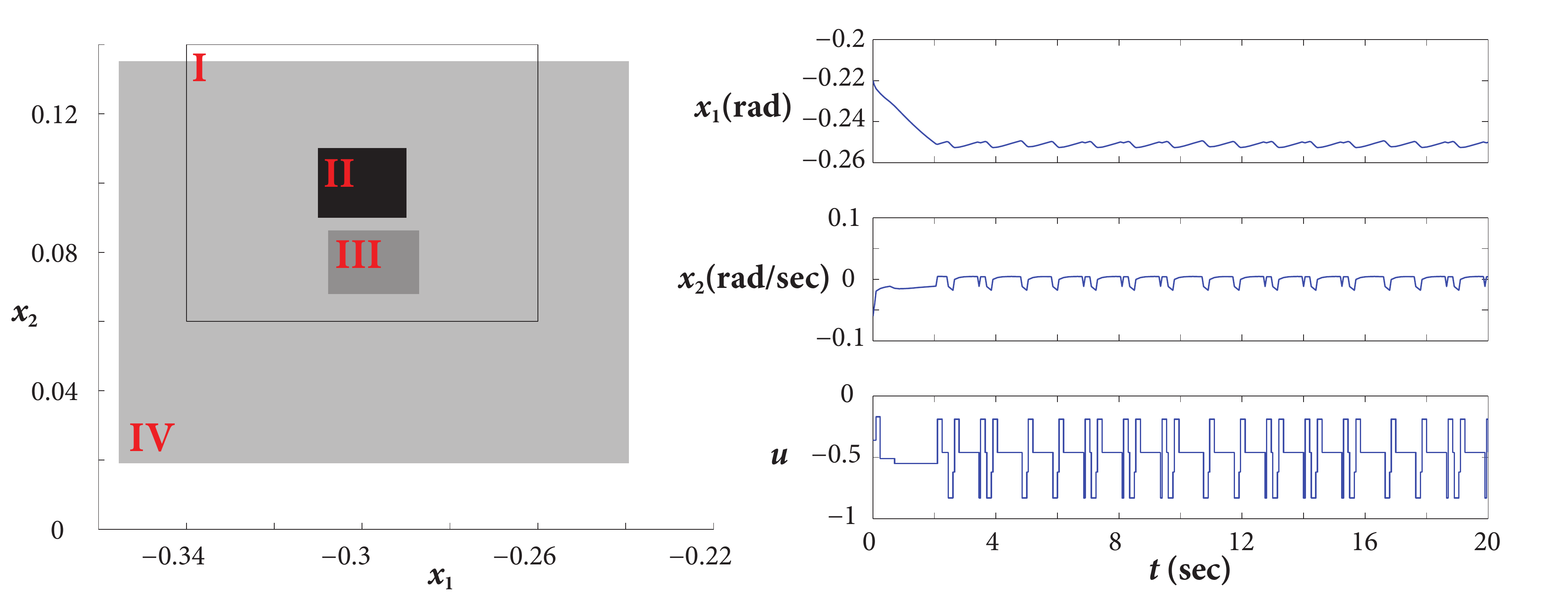}
      \caption{Left: The trajectories of the controlled pendulum system states and the corresponding control signal. Right: Comparison of one-step reachable sets generated by two methods: region I indicates the linearization region $\mathcal{B}_r(\hat{q})$; region II is the initial set of states; region III is an over-approximation of the reachable set obtained by the proposed linearization-based method; region IV is an over-approximation of the reachable set obtained by an analytical bound using Lyapunov-based methods.}
      \label{fig:pdl}
\end{figure}

\subsection{Automatic cruise control}

Consider the longitudinal dynamics of automatic cruise control
\begin{equation*}
\dot{v}=u-c_0-c_1v^2,
\end{equation*}
where $v \in [20,\,30]$, $u \in [-1.5,\,1]$, $c_0=0.1$, and $c_1=0.00016$.

To design a controller satisfying the specification $\varphi=\square(v \leq 30) \wedge \lozenge \square (v \in [22,\,24])$, we set $\tau_s=0.3 \text{s},\,r=0.6,\,\eta=0.1,\,\mu=0.2$. In the simulations, the system is subjected to a maximum delay $d=0.01$s and a measurement error bound $\varepsilon=0.1$m/s. We construct three different abstractions: i) one without robustness margins; ii) one with uniform robustness margins (as defined in \cite{LiuO14}); iii)  one with varying robustness margins (as defined in this paper). Fig.~\ref{fig:cruise} presents the simulation results of the cruise control system, under controllers synthesized using the first and the third abstractions, respectively. As observed from Fig.~\ref{fig:cruise} (left), the speed jumps out of the target range as the time lapses because the first abstraction cannot counteract delays or measurement errors, while the result from the third abstraction shown on the right of Fig.~\ref{fig:cruise} is satisfactory. To compare the second and the third abstractions, we look at their transitions around the state $v=21.4$m/s under the control input $u=0.15$. The second abstraction has $30$ transitions, whereas the third one has only $20$. In fact, due to its greater conservatism, the second abstraction is not able to generate a controller during control synthesis.

\begin{figure}[thpb]
      \centering
      \includegraphics[scale=0.4]{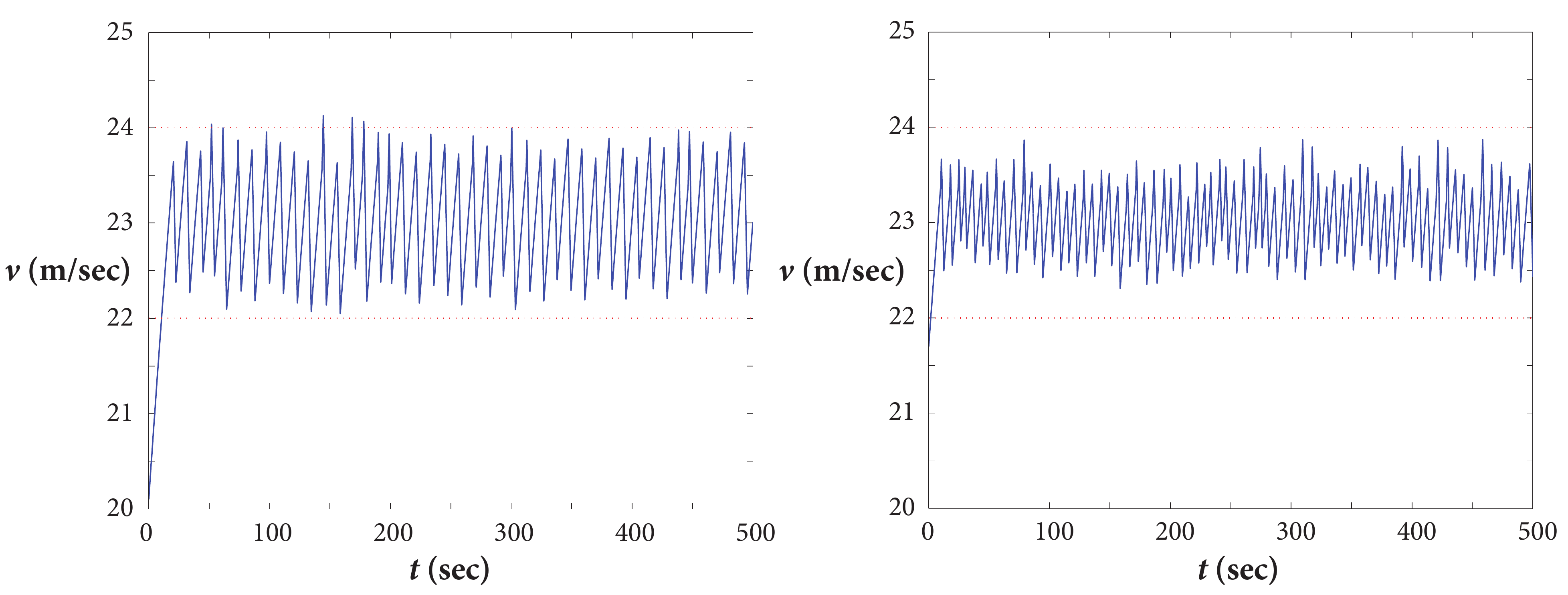}
      \caption{Controlled state evolution synthesized from an abstraction with (right) and without (left) local robustness margins.}
      \label{fig:cruise}
\end{figure}

\section{Conclusion}

In this paper, we considered the problem of constructing finite abstractions for nonlinear systems that are suitable for synthesizing robust controllers. A notion of finite abstractions with robustness margins that vary with respect to the local dynamics was formally defined. One main contribution of our work was to apply local reachable sets computation techniques in computing finite transitions, which led to reduced degree of nondeterminism in the abstractions. The local reachable sets are computed by linearization and approximation error estimation. As illustrated by numerical examples, the abstractions generated by the proposed method contain fewer spurious transitions than those obtained from Lyapunov-based methods and therefore are more likely to render the control synthesis problem realizable. Future work will combine the abstraction procedures presented in this paper, which take into account local dynamics, with automated refinement procedures to mitigate potential state explosion problem.

\bibliographystyle{IEEEtran}
\bibliography{rbabst}             
\end{document}